\documentclass[aps,pra,twocolumn,superscriptaddress,nofootinbib,a4paper,longbibliography]{revtex4-2}
\usepackage{soul}
\usepackage{bm,color,bbm}

\usepackage[colorlinks=true,linkcolor=blue,citecolor=magenta,urlcolor=blue]{hyperref}
\usepackage{mathtools,graphicx,soul}

\usepackage{subfigure}

\newcommand{\beq}{\begin{equation}}
\newcommand{\eeq}{\end{equation}}
\newcommand{\bqa}{\begin{eqnarray}}
\newcommand{\eqa}{\end{eqnarray}}

\newcommand{\ket}[1]{ |{#1} \rangle}

\definecolor{maroon}{rgb}{0.7,0,0}

\definecolor{ngreen}{rgb}{0.3,0.7,0.3}

\definecolor{golden}{rgb}{0.8,0.6,0.1}

\definecolor{npurple}{rgb}{0.3,0,0.6}

\usepackage{amsmath,amssymb,amsthm}
\usepackage{amsmath,blkarray}
\usepackage{braket}
\usepackage{color}
\usepackage{mathtools}
\usepackage{natbib}
\usepackage{slashbox}
\usepackage{cancel}
\allowdisplaybreaks

\newcommand{\ketbra}[2]{|#2\rangle\langle#1|}

\newcommand*{\Tr}{\mathrm{Tr}}

\newcommand{\ba}{\begin{eqnarray}}
\newcommand{\ea}{\end{eqnarray}}

\newtheorem{definition}{Definition}
\newtheorem{thm}{Theorem}
\newtheorem{cor}{Corollary}
\newtheorem{lem}{Lemma}

\begin{document}

\title{Operational Coherent Measurements with Steering and Randomness}

\author{Chellasamy Jebarathinam}

 \affiliation{Physics Division, National Center for Theoretical Sciences, National Taiwan University, Taipei 106319, Taiwan}

 \author{Huan-Yu Ku}
\email{huan.yu@ntnu.edu.tw}
\affiliation{Department of Physics, National Taiwan Normal University, Taipei 116059,  Taiwan}

   \author{Hsi-Sheng Goan}
\email{goan@phys.ntu.edu.tw}
 \affiliation{Department of Physics and Center for Theoretical Physics,
National Taiwan University, Taipei 106319, Taiwan}
\affiliation{Center for Quantum Science and Engineering, National Taiwan University, Taipei 106319, Taiwan}
 \affiliation{Physics Division, National Center for Theoretical Sciences, National Taiwan University, Taipei 106319, Taiwan}

\begin{abstract}
Measurement incompatibility underpins randomness generation in nonlocal phenomena. However, at its root, a more fundamental quantum feature is noncommuting (or coherent) measurements. This raises a central question: How can we operationally characterize the quantum advantage of coherent measurements within nonlocal correlations? We answer this by demonstrating that coherent measurements can leverage semi-device-independent (SDI) steering, enabling local randomness generation from any set of coherent measurements. Specifically, we establish that a measurement assemblage can be used to demonstrate SDI steering if and only if it is coherent, providing a complete operational characterization. To provide an application of this operational characterization, we formulate a nonconvex resource theory for SDI steering and propose an operational monotone for the two-setting scenario by mapping noncommuting measurements to SDI steering. Our framework enables a practical quantum random number generator based on SDI steering that eliminates the need to certify entanglement and tolerates arbitrarily low detection efficiency. That is, we demonstrate that genuine randomness can be generated via coherent measurements beyond standard steerable states and even beyond entangled states under realistic experimental conditions. These results extend the scope of quantum resources for generating nonlocal correlations beyond measurement incompatibility, revealing the operational power of coherent measurements.
\end{abstract}

\maketitle 
   \date{\today}
   
\textbf{\textit{Introduction}}.---  One of the fundamental aspects of quantum mechanics that serves to distinguish between classical and quantum systems is the noncommutativity of two or more observables, i.e. the Heisenberg uncertainty principle for projective measurements~\cite{Hei27}.
Instead, quantum mechanics allows non-projective measurements, called generalized measurements, or formally POVMs.  
For a set of POVMs, the measurement incompatibility is studied as a crucial resource for nonlocal correlations~\cite{HMZ16,HsiehPRA2025}, i.e. Bell nonlocality, contextuality~\cite{CarlesPRL2022,Kieran2022PRXQ} and quantum steering~\cite{GHK+23}. 
Among these correlations, quantum steering provides an operational interpretation of measurement incompatibility in the sense that a set of POVMs is incompatible if and only if it can demonstrate quantum steering~\cite{UMG14,QVB14,UBG+15}, and even quantify measurement incompatibility~\cite{Zhao2020npj,KHC+22,HKB23,Hsieh2023arXiv,Lee2025npj}. 
Here, quantum steering refers to an asymmetric nonlocal correlation, where Alice can spatially steer Bob's state by her local measurements~\cite{Sch35,WJD07}. 

However, the profound noncommutativity of measurements, which defines coherent measurements, is a more fundamental property than the measurement incompatibility of generalized measurements. In other words, the compatibility is genuinely inequivalent to the commutativity or incoherence of the measurements. In this context, noncommuting or coherent measurements that are compatible cannot be detected under a standard steering scenario. Therefore, a precise characterization of coherent measurements via nonclassical correlations remains an open question, as not all coherent measurements can demonstrate quantum steering. This fact limits the application of coherent measurements in the generation of randomness ~\cite{PCS+15,PaulPRL2018}.

In this work, we address the challenges of characterizing coherent measurements through nonlocal correlations and the limitations of practical applications under a quantum steering scenario. We resolve these two issues by additionally assuming the dimensionality of the quantum system on the steering party's side in the given steering scenario.  This scenario is called an SDI steering task. Due to this extra assumption, classical randomness is no longer freely available within the framework, leading to SDI steering as the nonconvex resource~\cite{JD23,JDS_PRA23,JDK+25,JKC+24}.
We show that a measurement assemblage can be used to demonstrate SDI steering if and only if it is coherent. Based on this demonstration, we next obtain a nonconvex monotone of SDI steering in the two-setting scenario. Finally, we demonstrate the operational application of this monotone in bounding the optimal guessing probability in a quantum random number generator (QRNG).
Thus, the practical advantages of using SDI steering as a resource for QRNGs are clearly illustrated. This approach demonstrates that genuine randomness can be generated via the coherence of measurements, eliminating the need to certify entanglement and tolerating any low detection inefficiency.

\textbf{\textit{Preliminaries}}.--- In quantum mechanics, a measurement is generally described by a positive-operator-valued-measure, i.e. a finite set $\{M_a\}_a$ of operators $0 \le M_a \le \openone_d$, where  $\openone_d \in \mathbb{C}^{d}$ is the identity operator such that $\sum_a M_a=\openone_d$. A set of POVMs with outcome $a$ for different settings $x$ is known as a measurement assemblage $\mathcal{M}^{n_a}_{n_x}:=\{M_{a|x}\}_{a,x}$, where $n_x$ and $n_a$ denote the number of setting and outcome, respectively. Such a measurement assemblage has pairwise mutual commutativity  if and only if 
$[M_{a|x},M_{a'|x'}]=0~\forall~a,x,a',x'.$
Otherwise, it has pairwise noncommutativity. 
 In case of projective measurements, i.e. $ M_{a|x} M_{a'|x}=\delta_{aa'} M_{a|x} $ for all $x$, the aforementioned condition of noncommutativity is associated with an uncertainty relation~\cite{Hei27,Rob29}.

For a commuting measurement assemblage, all elements in the set can be simultaneously diagonalizable in a single basis $\{\ket{i}\}$, i.e. 
\begin{equation} \label{icas}
M_{a|x}= \sum^{d-1}_{i=0} \alpha_{i|(a,x)} \ketbra{i}{i},
\end{equation}
with $\alpha_{i|(a,x)}=\braket{i|M_{a|x}|i}$.
Thus, a measurement assemblage $\mathcal{M}^{n_a}_{n_x}$ is termed incoherent if it admits a description in a single orthonormal basis, i.e. when each of its elements has a diagonal matrix representation~\cite{TKK+23}. This is because all POVM constituent elements can be diagonalizable within the computational basis $\{\ket{i}\}_i$. Otherwise, it is said to have coherence. 

The other type of nonclassicality of the measurement assemblage is the measurement incompatibility, captured through nonjoint measurability~\cite{GHK+23}.
Specifically, a measurement assemblage is jointly measurable if there is a single 
POVM $\{G_\lambda\}_\lambda$ such that each element of the assemblage can be obtained as 
\begin{equation}
\label{Eq:JM_model}
M_{a|x}= \sum_\lambda p(a|x,\lambda) G_\lambda,
\end{equation}
where $p(a|x,\lambda)$ is a conditional probability. It is obvious to see that any incoherent measurement assemblage is compatible; however, the converse is not true because the coherence of the given measurement assemblage does not necessarily imply incompatibility~\cite{UMG14}. Evidently, admission of Eq.~\eqref{icas} by a measurement assemblage implies admission of Eq.~\eqref{Eq:JM_model}; nevertheless, this implication is not reversible. This follows because the set of incoherent assemblages of the given $d$-dimensional quantum system forms a nonconvex set, whereas the compatible set has convexity.

Any incompatible measurement assemblage has an operational characterization via quantum steering~\cite{QVB14,UBG+15}. This characterization is demonstrated by a steering protocol~\cite{EPR35, Sch35, WJD07} where Alice and Bob share a bipartite quantum state $\rho_{AB}$. To steer Bob's reduced state into different ensembles, Alice performs measurements on her subsystem. Through this measurement assemblage, she prepares the resulting state assemblage $\mathcal{S}^{n_a}_{n_x}:=\{\sigma_{a|x}\}_{a,x}$ for Bob.
Each element in the state assemblage is given by the set of unnormalized conditional states
\begin{equation}
\sigma_{a|x}={\rm Tr}_A{( (M_{a|x} \otimes \openone) \rho_{AB})} \hspace{0.5cm} \forall \sigma_{a|x} \in \mathcal{S}^{n_a}_{n_x},
\end{equation}
which gives the conditional probability of obtaining Alice's outcome as  
$p(a|x)=\rm Tr (\sigma_{a|x})$ and the reduced state of Bob as $\rho_{B}=\sum_a\sigma_{a|x}$. A state assemblage is unsteerable if it has a local-hidden-state (LHS) model~\cite{WJD07}, i.e.
 for all $a$, $x$, there is a decomposition of $\sigma_{a|x}$ in the form
\begin{equation}
\sigma_{a|x}=\sum_\lambda p(\lambda) p(a|x,\lambda) \rho_\lambda, 
\end{equation}
where $\lambda$ denotes the classical random variable that occurs with probability 
$p(\lambda)$; $\rho_{\lambda}$'s satisfy $\rho_\lambda\ge0$ and
${\rm Tr}(\rho_\lambda)=1$. Otherwise, it is called steerable. Steerability can be operationally identified as the certification of entanglement of the shared state in a one-sided device-independent (1SDI) way or the certification of measurement incompatibility at the untrusted side~\cite{WJD07,UBG+15}. Besides, quantum steering can be applied as a resource to 1SDI QRNGs~\cite{PCS+15,PaulPRL2018}.  

There is also a strong connection between the incompatibility of any measurement assemblage and steerability. The key notion to establish this correspondence is steering-equivalence-observables (SEO), which map a state assemblage to steering-equivalent POVMs via
\begin{equation}\label{SEO}
B_{a|x}= \rho^{-1/2}_B \sigma_{a|x} \rho^{-1/2}_B \quad \texttt{if} \quad \rho_B \quad \texttt{is full-rank}.
\end{equation}
For a non-full-rank $\rho_{B}$, an analogous expression should be written with an isometry that maps $\rho_{B}$ to a subspace where it has support. 
This connection establishes a one-to-one correspondence: a measurement assemblage is compatible if and only if it cannot demonstrate steerability, while an incompatible measurement assemblage can always demonstrate it~\cite{QVB14, UBG+15}. 

Our objective is to provide an operational characterization of the coherence of any given measurement assemblage in a nonlocal scenario.
To this end, we consider steering-inspired tasks in the one-sided semi-device-independent (1SSDI) context~\cite{JDS_PRA23,JKC+24} instead of the 1SDI context of steering tasks observing Einstein-Podolsky-Rosen steering~\cite{EPR35,Sch35,WJD07}.
The difference between these two contexts is that the Hilbert space dimension on the untrusted side is assumed in the 1SSDI context, whereas the untrusted side is fully device-independent in the 1SDI context. This dimensional assumption resembles that of the prepare-and-measure scenarios, where preparation is trusted, but measurements are not characterized to provide SDI applications~\cite{BBP+26}. For example,  these applications include SDI certification of quantum measurements~\cite{TKV+18}. In the case of steering-inspired tasks, the 1SSDI context is therefore relevant in scenarios where preparation can be trusted to demonstrate the phenomenon. 
In the present work,  we provide SDI characterization of the coherence of measurements in the 1SSDI context.  

As a consequence of the dimensional assumption of the quantum state in the 1SSDI context, the hidden variable $\lambda$ in witnessing steerability also has the same dimensional restriction. 
The SDI steering~\cite{JDS_PRA23}, which is a steering-inspired task, is defined as follows.
\begin{definition}\label{Defstrun}
 Suppose that a state assemblage $\mathcal{S}^{n_a}_{n_x}$ is produced in the given 1SSDI scenario. We say that SDI steering is demonstrated if and only if the state assemblage does not admit a decomposition of the form
\begin{equation}\label{LHSdl}
\sigma_{a|x}=\sum^{d_\lambda-1}_{\lambda=0} p(\lambda) p(a|x,\lambda) \rho_\lambda \hspace{0.5cm} \forall \sigma_{a|x} \in \mathcal{S}^{n_a}_{n_x},
\end{equation}
with $d_\lambda \le d_A$, where $d_A$ is the dimension of Alice's system, over all such decompositions.
\end{definition} 
All steerable states in the 1SDI context can obviously be used to demonstrate SDI steering. 
 The dimensional assumption in the definition of SDI steering relaxes the convexity constraint of the standard steering~\cite{JD23,JKC+24}; in other words, operationally, the dimensional assumption implies that when using SDI steering as a resource, shared classical randomness is not freely available.  Also note that for the experimental context of observing steering, classical randomness can be a resource. Classical randomness as a resource was used in~\cite{GPW05,LBL+15}, providing an operational  motivation to identify quantum resources beyond convexity for quantum technologies~\cite{KTA+24,JKC+24,SCR+25}.
The 1SSDI context for observing steering with classical randomness as a resource is analogous to the prepare-and-measure context with independent preparation and measurement devices as employed in~\cite{LBL+15} to provide an SDI QRNG.


\textbf{\textit{Operational coherence of measurements}}.--- 
Here, we demonstrate that all coherent measurements can be operationally characterized via SDI steering.
To do this, we start with the following theorem for any SEO $\mathcal{SO}^{n_a}_{n_x}:=\{B_{a|x}\}_{a,x}$ in the 1SSDI context.
\begin{thm} \label{1-to-1}
Any given SEO  $\mathcal{SO}^{n_a}_{n_x}$ is incoherent if and only if the state assemblage admits a dimensionally-restricted LHS model as in Eq.~(\ref{LHSdl}) with $d_\lambda \le d_A$. 
\end{thm}
The proof of this theorem is given in App.~\ref{app:thm1}.

Next, note that any incoherent measurement assemblage can only be used to obtain an incoherent SEO for any state $\rho_{AB}$.  
In this context, the state assemblage produced on Bob's side using an incoherent measurement assemblage as in Eq.~(\ref{icas}) is given by 
\begin{equation}
\sigma_{a|x}= \sum_{i} \alpha_{i|(a,x)} \mathrm{Tr}_A \left(\left(\ketbra{i}{i} \otimes \openone \right) \rho_{AB} \right),
\end{equation}
which is an LHS model with the dimension of the hidden variable bounded by $d_\lambda \le d_A$. Thus, any incoherent measurement assemblage always leads to incoherent SEO.
On the other hand, any measurement assemblage $\mathcal{M}^{n_a}_{n_x}$ that has coherence can be used to obtain a state assemblage whose SEO has coherence.
To see this, consider any pure entangled state in the Schmidt decomposition $\ket{\phi^{(d)}_+}=\sum^{d-1}_{i=0} \lambda_i \ket{ii}.$
Any state assemblage arising from a pure entangled state is given by
\begin{equation}\label{SApent}
\sigma_{a|x}= \rho^{1/2}_B M^T_{a|x}  \rho^{1/2}_B,
\end{equation}
where  $M^T_{a|x}$ is the transpose of $M_{a|x}$.
Then, the SEO of the state assemblage that has the above decomposition is given by 
$B_{a|x}= M^T_{a|x}$,
which is incoherent if and only if the measurement assemblage is incoherent. 
Thus, as an implication of Theorem~\ref{1-to-1}, we have the following corollary.
\begin{cor} 
Suppose that the given state assemblage $\mathcal{S}^{n_a}_{n_x}$ is produced by the measurement assemblage $\mathcal{M}^{n_a}_{n_x}$. Then $\mathcal{S}^{n_a}_{n_x}$ admits a decomposition of a dimensionally-restricted LHS model, as in Eq.~(\ref{LHSdl}), for any state $\rho_{AB}$ if and only if $\mathcal{M}^{n_a}_{n_x}$ is incoherent.  
\end{cor}
In the above corollary, we have established that the coherence of any measurements can be detected by the coherence of SEO in the 1SSDI scenario. Thus,
the coherence of any measurements has SDI characterization through detection via SDI steering. To provide an application of this operational coherence of measurements to randomness, we are now going to establish a quantitative SDI steering based on the SEO's property in the following.

\textbf{\textit{Quantification of SDI steerability}}.---
As another implication of Theorem~\ref{1-to-1}, we have the following tight criterion for SDI steering.
\begin{cor} \label{SEONC}
SDI steering is demonstrated if and only if SEO $\mathcal{SO}^{n_a}_{n_x}$ has pairwise noncommutativity.
\end{cor}
\begin{proof}
A state assemblage admits a decomposition of the dimensionally-restricted LHS model with $d_\lambda \le d_A$  if and only if the SEO commute, i.e. $[B_{a|x}, B_{a'|x'}]=0$ for all $a,x,a',x'$. On the other hand, any SEO with noncommutativity implies that the state assemblage does not have the above-mentioned decomposition. 
\end{proof}

Based on the property of SEO in the 1SSDI context stated in Cor.~\ref{SEONC}, we now define a resource-theoretic quantification of the phenomenon. Before that, we need to formulate the resource theory of steerability in the 1SSDI context. In this resource theory, the free resources are all state assemblages that have the dimensionally-restricted LHS model with $d_\lambda \le d_A$; while the standard 1SDI does not admit this restriction~\cite{GallegoPRX2015,KuPRA2018,KuPRXQ2022}. On the other hand, in the following theorem, we identify the free operations.
\begin{thm}\label{fo}
The free operations of the resource theory of SDI steering are any local operations without shared randomness of the form given by
\begin{equation}\label{FOcSEO}
\sigma'_{a'|x'}=\sum_{a,x,\mu} p(\mu) p(a'|a,x',\mu) p(x|x',\mu)
 \mathcal E(\sigma_{a|x}),
\end{equation}
where the new state assemblage $\{\sigma'_{a'|x'}\}_{a',x'}$ is obtained from the original state assemblage $\{\sigma_{a|x}\}_{a,x}$, $p(\mu)$, $p(a'|a,x',\mu)$ and $p(x|x',\mu)$ are probabilities and  $\mathcal E$:  $\mathbb{C}^{d} \rightarrow \mathbb{C}^{d}$, is the completely positive trace-preserving  (CPTP) map. 
\end{thm}
The proof of this theorem is given in App.~\ref{app:thm2}. 
 Within the framework of local operations without shared randomness as stated in Theorem~\ref{fo}, the term \textit{nonlocal correlations} is used in the current work. In this broader framework,  nonlocal correlations include both SDI steerable and standard steerable correlations.
Now, we construct a proper quantification of SDI steerability in the context of the above resource theory. In a two-setting 1SSDI scenario, a measure based on the amount of SEO's noncommutativity can be captured through a Schatten $p$-norm $||A||_p$ of an operator $A$. Let $\sigma_j$ denote the singular values of $A$. $||A||_p$ is defined as
\begin{align*}
    \| A \|_p :=
    \begin{cases}
      \big( \sum_j \sigma_j^p \big)^\frac{1}{p}, \; &\text{for} \; p \in [1,\infty[,\\
      \max_j \, \sigma_j, \; &\text{for} \; p=\infty,
    \end{cases}
\end{align*}
where the summation and the maximum go over all the singular values of $A$ (including multiplicities).
For any given measurement assemblage $\mathcal{M}^{n_a}_{2}$, the amount of noncommutativity  can be defined as follows: 
\begin{equation} \label{SEOcomm}
\Upsilon_p(M_{a|x})=\sum_{a,a'}||[M_{a|0},M_{a'|1}]||_p.
\end{equation} 
The above measure has been studied in~\cite{MK22}.
For any measurement assemblage $\mathcal{M}^{n_a}_{2}$ in $\mathbb{C}^{d}$, the measure of noncommutativity satisfies $0 \le \Upsilon_p \le 2^{\frac{1}{p}} d \sqrt{d-1}$. Note that $\Upsilon_p$ is concave.

 Using the measure of noncommutativity of  SEO as in Eq.~(\ref{SEOcomm}), we then define the quantification of SDI steerability $S_{\Upsilon_p}$ as follows. 
\begin{definition} \label{QSDIS}
For any state assemblage $\mathcal{S}^{n_a}_{2}$ produced in the given two-setting  1SSDI scenario, the corresponding SEO $\mathcal{SO}^{n_a}_{2}$ can be constructed. Then, we define the measure of SDI steerability $S_{\Upsilon_p}(\sigma_{a|x})$ as
\begin{equation} \label{QSDISt}
    S_{\Upsilon_p}(\sigma_{a|x})=\frac{\Upsilon_p(B_{a|x})}{2^{\frac{1}{p}} d \sqrt{d-1}},
\end{equation}
which satisfies  $0\le S_{\Upsilon_p}(\sigma_{a|x})\le 1$.
\end{definition}
The above quantification is a faithful, nonconvex, and monotonic measure of SDI steering. That is, the measure in Eq.~(\ref{QSDISt}) satisfies the following axiomatic properties  in resource theory~\cite{ChitambarRMP2019} for any $p$. We now formally state the properties of our measure and prove them.
\begin{enumerate}
    \item Faithfulness: $ S_{\Upsilon_p}(\sigma_{a|x}) = 0$ if and only if the state assemblage is a free resource.
\item Nonconvexity: for any state assemblages $\mathcal{S}^{n_a}_{2}$ and $\mathcal{S'}^{n_a}_{2}$,  $S_{\Upsilon_p}(q\sigma_{a|x}+(1-q)\sigma'_{a|x}) \ge  q S_{\Upsilon_p}(\sigma_{a|x})+(1-q)S_{\Upsilon_p}(\sigma'_{a|x})$ for $0 \le q \le 1$. 

\item Monotonicity: Let $\mathcal{F}$ denote a free operation as in Eq.~(\ref{FOcSEO}). 
$ S_{\Upsilon_p}(\mathcal{F}[\sigma_{a|x}]) \le  S_{\Upsilon_p}(\sigma_{a|x})$.
\end{enumerate}
The first property holds because $S_{\Upsilon_p}(\sigma_{a|x}) = 0$ if and only if the state assemblage is a free resource in the 1SSDI context, due to Cor.~\ref{SEONC}. The second property is valid due to the concavity property of the noncommutativity measure~\cite{MK22}.  The proof of the third property is given in App.~\ref{monotone}. In addition, in App.~\ref{illu}, using our measure $ S_{\Upsilon_p}(\sigma_{a|x})$, the quantification of SDI steering of specific two-qudit states and the detection of SDI steering in the presence of any low detection efficiency have been illustrated. 
Next, we aim to relate our measure of SDI steering  to the figure of merit of a QRNG for which the coherence of the measurements is the resource.

\textbf{\textit{Characterization of SDI randomness}}.---  Quantum randomness arises through the coherence complementary relation between state and measurement. For example, measuring a single qubit state $\ket{\psi}=\frac{\ket{0}+\ket{1}}{2}$ in the basis $\{\ket{0},\ket{1}\}$ breaks the coherence in the state and provides completely random outcomes. This randomness is fully unpredictable. The intrinsic randomness of quantum systems is exploited as a resource in QRNGs~\cite{MYC+16,HGJ17,MVP23} and quantum cryptography~\cite{PR22}. Obviously, QRNG based on quantumness of single systems requires a trust in both the state and the measurement so that the coherence is exploited~\cite{MYC+16}. On the other hand, entanglement-based QRNGs can be used to provide certifiable randomness even if devices are not trusted, such as 1SDI randomness certification~\cite{PCS+15}. 
In this context, intrinsic randomness is contained in $\{p(a|x)\}_{a,x}$.  
This also holds in the presence of an eavesdropper, Eve, who can hold a purification $\ket{\psi_{ABE}}$ of the state shared by Alice and Bob $\rho_{AB}$, with the dimension of Eve's system being arbitrary~\cite{PCS+15}. However, this randomness in a nonlocal scenario is based on measurement incompatibility rather than the coherence of measurement, cf. a single system randomness scenario~\cite{VMT+14}. 

Now we close this fundamental gap in the 1SSDI context, such that the randomness can be generated in a nonlocal scenario with coherent measurements, namely:
\begin{thm}\label{IRandom}
 Suppose $\{p(a|x)\}_{a,x}$ is observed in the given 1SSDI scenario with SEO having noncommutativity. In that case, intrinsic randomness is contained in $\{p(a|x)\}_{a,x}$ even in the presence of an adversary, Eve, but whose system is also dimensionally restricted as Alice's.
 \end{thm}
 The proof of Theorem~\ref{IRandom} is given in App.~\ref{app:thm3}. Several physical insights follow from the above theorem. First, the intrinsic randomness of $\{p(a|x)\}_{a,x}$ is certified by SDI steering, provided that Eve's purification, $\ket{\psi_{ABE}}$, has the limitation that her dimension is limited to that of Alice. Second, the above Theorem shows that coherence of measurements can be applied to randomness generation in the nonlocal scenario.

\begin{figure}
	\begin{subfigure}
		\centering\includegraphics[width=8cm]{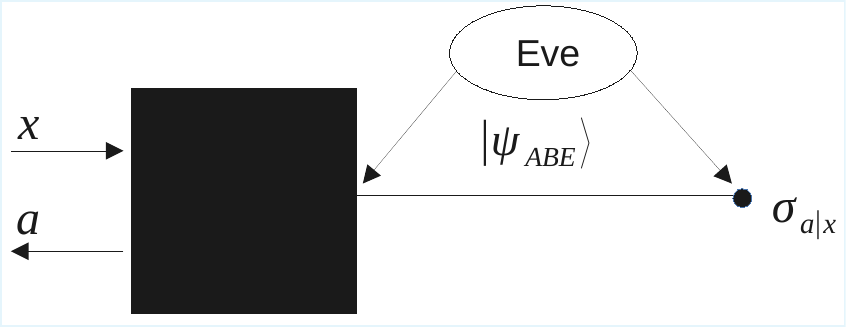}
		\caption{Setup for randomness certification in a 1SDI scenario, where Alice treats her measurement device as a black box with inputs and outputs, labeled by $x$  and $a$, respectively, and Eve's role can be described by having a purification $\ket{\psi_{ABE}}$ of the state $\rho_{AB}$ shared by Alice and Bob~\cite{PCS+15}. Here, the local dimensions $d_A$ and $d_E$ of Alice's and Eve's systems, respectively, are arbitrary, while the local dimension $d_B$ of Bob's system is trusted.    
        \label{Fig:QRNG_1sDI}}
	\end{subfigure}
	\begin{subfigure}
		\centering\includegraphics[width=8cm]{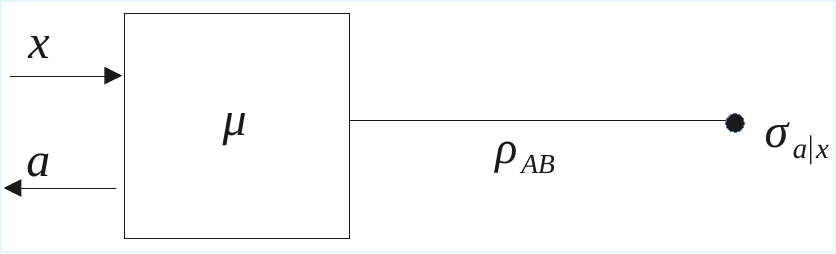}
		\caption{Setup for randomness certification in the 1SSDI scenario, where only the dimension of Alice is trusted, but she performs uncharacterized measurements. Here, the fully trusted party Bob prepares the bipartite state $\rho_{AB}$ and sends one half of it to Alice. Eve's role can be described by local randomness $\mu$, which accounts for noise acting on Alice's side.  \label{Fig:QRNG}}
	\end{subfigure}
	\end{figure}
In Fig.~\ref{Fig:QRNG_1sDI}, the setup for randomness certification in a 1SDI scenario is depicted, while the setup for randomness certification in the 1SSDI scenario is depicted in Fig.~\ref{Fig:QRNG}.
In the operational context of certifying randomness,  the difference between the two scenarios is that in any 1SDI scenario, Eve, who is not dimensionally restricted, can prepare the bipartite quantum state shared by Alice and Bob, whereas in any 1SSDI scenario, the trusted party, Bob, prepares the bipartite quantum state to be shared with Alice. In the case of 1SSDI scenario, Eve's presence can be modeled as noise which is described by an internal source of randomness acting on Alice's device, as in the prepare-and-measure context of~\cite{LBL+15}.

Now, we show that the relationship between SDI QRNG and SDI steerability is even tighter.
Specifically, we demonstrate the quantitative relation between the amount of intrinsic randomness in $\{p(a|x)\}_a$ for the given $x$ with two dichotomic measurements and the measure of SDI steerability, $S_{\Upsilon_p}$, given by Definition~\ref{QSDIS}.  We restrict ourselves to $d_A=2$ in the following because the required mathematical tools are not available for higher dimensional cases to establish the aforementioned tighter relationship.

Using SDI steerability as a resource for randomness generation at the untrusted side, Eve cannot use shared randomness between Alice and Bob, and her purification has the limitation stated in Theorem~\ref{IRandom}. Thus, Eve's role can be modeled as a source of internal randomness $\mu$ at the untrusted side~\cite{LBL+15,JD23}. Then, $\{p(a|x)\}_{a,x}$ observed in the 1SSDI scenario is given by 
\begin{align}
p(a|x)&=\int d\mu r(\mu) p(a|x,\mu)
      =\Tr(\rho_{AB} (M_{a|x} \otimes \openone_2) ),
\end{align}
where $\mu$'s occur with probability distribution $r(\mu)$ and $\rho_{AB}$ is the two-qubit state shared by Alice and Bob. Here, the internal randomness describing the noise at the level of measurements is as follows: 
\begin{equation}
M_{a|x}=\int d\mu r(\mu) M^{(\mu)}_{a|x}= \frac{1}{2}\left(\openone_2 + (-1)^a\vec{T}_x \cdot \vec{\sigma} \right),
\end{equation}
where $\vec{T}_x$ is the Bloch vector of Alice's measurement and $\vec{\sigma}$ is the vector of Pauli matrices. Similarly, at the level of the state assemblage, we have
\begin{equation}
\sigma_{a|x}=\int d\mu r(\mu) \sigma_{a|x}^{(\mu)}.
\end{equation}

The figure of merit of the random number generation protocol is 
the best guessing probability, i.e. $p_{g}=\max_a p(a|x)$,  for the given $x$ and the state assemblage produced at Bob's side. 
Given $x$ and the knowledge of the internal state $\mu$ that occurs with probability $q_\mu$, the best guess for $a$ is given by $\max_a p(a|x, \mu)$. 
The upper bound on the figure of merit is then given as 
\begin{align}
p_{g}&=\sum_{\mu} q_\mu \max_a  p(a|x,\mu) 
        \le  \frac{1}{2} \left(1+\sqrt{1-S^2_{\Upsilon_p}}\right).
\label{guessub}
\end{align}
The proof of the above inequality is given in App.~\ref{proofSTQRNG}.

Inequality~(\ref{guessub}) implies that maximal local randomness, i.e. $H_{\min}:=-\log_2 p_g=1$ is certified using any pure entangled state. This follows because $S_{\Upsilon_p}=1$ can be achieved by any pure entangled state (see Eq.~(\ref{QSck}) in App.~\ref{illu}). For the two-qubit isotropic states, $\rho^{(2)}_{iso}=\alpha \ketbra{\phi^{(2)}_+}{\phi^{(2)}_+} +\frac{1-\alpha}{4} \openone_4$, the certified randomness using SDI steerability of the state is given by $H_{\min}=-\log_2 \left(\frac{1}{2} \left(1+\sqrt{1-\alpha^2}\right)\right)$, which is nonzero for any $|\alpha|>0$. Thus, SDI local randomness can also be achieved by an entangled state that is unsteerable in the standard steering context or by a separable state.  Such SDI local randomness that does not require using standard steering or entanglement  can also be obtained  using prepare-and-measure correlations, as demonstrated in~\cite{LBL+15}. Our framework demonstrates that the quantum communication advantage of the protocol in~\cite{LBL+15}  can also be achieved by our operational coherence of measurements via SDI steering. 

Finally, we compare the framework of the certification of randomness studied by us  with that of the standard steering. First, in~\cite{PaulPRL2018}, it was demonstrated that maximal randomness can be witnessed through the optimal violation of the steering inequality using a two-qudit system. Otherwise, they use a semi-definite program to estimate the randomness in a noisy scenario. Here, we demonstrate an analytical relation directly through the quantification of SDI steering, as in Eq.~\eqref{guessub}. In addition, incompatible measurements and detection efficiency higher than a certain threshold value are required for randomness generation in the standard steering scenario~\cite{PCS+15}. On the other hand, SDI local randomness indicated by a nonzero $S_{\Upsilon_p}$ can also be achieved by the isotropic two-qubit states for noisy coherent  mutually unbiased bases, which are not necessarily incompatible. In App.~\ref{illu}, we have shown that in the presence of any low detection efficiency, SDI local randomness can still be achieved by the isotropic two-qubit state for any $|\alpha|>0$.

\textbf{\textit{Conclusions and Discussions}}:--- We have established that coherent measurements can be operationally characterized through SDI steering, proving that a measurement assemblage demonstrates SDI steering if and only if it is coherent. This fundamentally differs from standard steering, as coherent but compatible measurements, which cannot be used to demonstrate standard steering, can nevertheless be used to demonstrate SDI steering. By formulating a nonconvex resource theory and introducing an operational monotone, we provided quantification of SDI steerability and applied it to design a QRNG that eliminates the need to certify entanglement and tolerate arbitrarily low detection efficiency. These results reveal that coherent measurements, the fundamental quantum feature underlying the Heisenberg uncertainty principle, serve as an operational resource for quantum information tasks, opening new avenues for SDI protocols under realistic experimental conditions. 

Finally, some questions remain open. 
To demonstrate the quantum advantage of SDI steering in general, it would be interesting to study a unified measure of SDI steering for any number of settings with any number of outcomes. 
We have restricted ourselves to the qubit case in designing our SDI QRNG. Thus, it would be relevant to develop the required mathematical tools and the semi-definite-programming method to extend our SDI QRNG to any higher dimension.
Reference~\cite{CMS+25} characterized intrinsic randomness from generalized measurements on single-qubit states, relevant to QRNG setups similar to ours. Connecting this characterization to our QRNG framework would be an interesting direction. 
 Since quantum randomness is a resource for secure key distribution, it would be highly relevant to connect our framework to achieve secure key distribution protocols.

\textbf{\textit{Acknowledgements}}:---This work was supported by the National Science and Technology Council (NSTC), the Ministry of Education (MOE) through the Higher Education Sprout Project NTU-113L104022-1, and the National Center for Theoretical Sciences (NCTS) of Taiwan.
 H.- Y. K. is supported by National Science and Technology Council (NSTC), (with grant number NSTC 112-2112-M-003- 020-MY3), and the MOE through Higher Education Sprout Project of National Taiwan Normal University.
 H.-S.G. acknowledges support from NSTC under Grants No. NSTC 113-2112-M-002-022-MY3, No. NSTC 113-2119-M-002 -021, and No. NSTC 114-2119-M-002-017-MY3, from the US Air Force Office of Scientific Research under Award Number FA2386-23-1-4052 and the support of Taiwan Semiconductor Research Institute (TSRI) through the Joint Developed Project (JDP), and the support from the Physics Division, NCTS. 
H.-S.G. acknowledges support from the research project ``Pioneering Research in Forefront Quantum Computing, Learning and Engineering'' of National Taiwan University (NTU) under Grant No. NTU-CC- 113L891604  No.~NTU-CC-114L8950, and No.~NTU-CC-114L895004, as well as the support from the “Center for Advanced Computing and Imaging in Biomedicine (NTU-113L900702 and NTU-114L900702)” through the Featured Areas Research Center Program within the framework of the MOE Higher Education Sprout Project.

\appendix

\section{Proof of Theorem~\ref{1-to-1}} \label{app:thm1}
Suppose that any given SEO $\mathcal{SO}^{n_a}_{n_x}$ is incoherent. Then there exists an orthonormal basis $\{\ket{i}_{B}\}$ on Bob's side and coefficients $0\le \beta_{i|(a|x)}\le 1$ such that 
\begin{equation}\label{icSEO}
B_{a|x}= \sum_{i} \beta_{i|(a|x)} \ketbra{i}{i},
\end{equation}
with $\beta_{i|(a|x)}=\braket{i|B_{a|x}|i}$. It then follows that the state assemblage does not imply SDI steering because it has the dimensionally-restricted LHS model, as mentioned in Definition~\ref{Defstrun}. To show this, we assume that Bob's reduced state is full rank. After plugging Eq. (\ref{icSEO}) into the definition of SEO in Eq.~(\ref{SEO}), we have
\begin{equation}
\sum_{i} \beta_{i|(a|x)} \ketbra{i}{i}=\rho^{-1/2}_B \sigma_{a|x} \rho^{-1/2}_B.
\end{equation}
From the above relation, it follows that the state assemblage admits the decomposition given by
\begin{equation}
\sigma_{a|x}=\sum_{i} \beta_{i|(a|x)}  \rho^{1/2}_B\ketbra{i}{i}  \rho^{1/2}_B.
\end{equation}
On the right-hand side of the above equation, we can rewrite the decomposition as follows 
\begin{equation}\label{strLHSic}
\sigma_{a|x}=\sum_{i} \beta_{i|(a|x)}  \sigma_i,
\end{equation}
by introducing the unnormalized states $\sigma_i=\rho^{1/2}_B\ketbra{i}{i}  \rho^{1/2}_B$.
Obviously, this is an LHS model with $i$ denoting a local hidden variable. This implies that the condition for having the dimensionally-restricted LHS model with $d_\lambda \le d_A$ is satisfied.

To prove the converse, consider any state assemblage that has the dimensionally-restricted LHS model given by 
\begin{equation}\label{MDILHS1to1}
\sigma_{a|x}=\sum^{d_\lambda-1}_{\lambda=0} p(a|x,\lambda) \sigma_\lambda, 
\end{equation}
with $d_\lambda \le d_A$ and $\sum_a p(a|x,\lambda)=1$. 
Note that in any such LHS model, we must have
\begin{equation}
\sum_\lambda \rho^{-1/2}_B \sigma_\lambda \rho^{-1/2}_B=\openone , \quad \forall \sigma_\lambda,
\end{equation}
which implies that 
\begin{equation}\label{slorb}
\sigma_\lambda=\rho^{1/2}_B\ketbra{\lambda}{\lambda}  \rho^{1/2}_B, 
\end{equation}
for the orthonormal basis $\{\ket{\lambda}\}$ that diagonalizes $\rho_B$. This can be seen as follows. Inserting $\sigma_\lambda$ given by Eq. (\ref{slorb}) in Eq. (\ref{MDILHS1to1}) and summing over $a$, we obtain 
\begin{align}
\sum_a\sigma_{a|x}&=\sum_a\sum^{d_\lambda-1}_{\lambda=0} p(a|x,\lambda) \rho^{1/2}_B\ketbra{\lambda}{\lambda}  \rho^{1/2}_B, \nonumber \\
&=\sum^{d_\lambda-1}_{\lambda=0}  \rho^{1/2}_B\ketbra{\lambda}{\lambda}  \rho^{1/2}_B  =\rho_B.
\end{align}
Now, for any state assemblage that has an LHS model as in Eq. (\ref{MDILHS1to1}), it is readily seen that the SEO has the following form,
\begin{equation}
B_{a|x}=\sum_\lambda p(a|x,\lambda) \ketbra{\lambda}{\lambda}, 
\end{equation}
which is incoherent.
This implies that any dimensionally-restricted LHS model with $d_\lambda \le d_A$ leads to an SEO that is incoherent, as in Eq.~(\ref{icSEO}). 

\section{Proof of Theorem~\ref{fo}}  \label{app:thm2}

To identify the free operations of the resource theory of SDI steering, we first recall that a resource-theoretic framework of steering was studied to characterize the nonclassicality of steerable assemblages in the 1SDI context~\cite{ZSHS23}. In this framework, the key element is the precise characterization of free operations in terms of local operations and shared randomness (LOSR). Under an LOSR transformation, a given state assemblage $\mathcal S^{n_a}_{n_x}$ transforms into another assemblage $\mathcal S'^{n_a}_{n_x}$, with the elements given by
\begin{align} \label{FOinSEO}
\sigma'_{a'|x'}=\sum_{\lambda,a,x} p(\lambda) p(a'|a,x',\lambda) p(x|x',\lambda)
p(a|x) \mathcal E_\lambda(\rho_{a|x}),
\end{align}
where $p(a'|a,x',\lambda)$ encodes the classical postprocessing of Alice’s output $a$, as a function of $x'$ and shared classical randomness and $\mathcal E_\lambda[\cdot]$ is the CPTP map corresponding to Bob’s local postprocessing of his quantum system, as a function of shared classical randomness $\lambda$. We now proceed to obtain the following lemma.
\begin{lem} \label{nFO}
Under an LOSR transformation of the form~(\ref{FOinSEO}), a state assemblage with commuting SEO can be transformed into one with noncommuting SEO. 
\end{lem}
\begin{proof}
To prove the statement, we demonstrate that under such an LOSR transformation given by Eq.~(\ref{FOinSEO}), a state assemblage, with an LHS model having $d_\lambda \le d_A$, can be transformed into another state assemblage, with an LHS model having a higher hidden variable dimension. To see this, let us apply a LOSR transformation of the form~(\ref{FOinSEO}) with $\lambda$ denoted as $\lambda'$ to an assemblage having an LHS model with  $d_\lambda \le d_A$  as follows:
\begin{multline}
\sigma'_{a'|x'} = \sum_{\lambda',\lambda,a,x} p(\lambda) p(\lambda') p(a'|a,x',\lambda') p(x|x',\lambda') \\
\cdot p(a|x,\lambda) \mathcal E_{\lambda'}(\rho_\lambda).
\end{multline}
We rewrite the above decomposition of $\{\sigma'_{a'|x'}\}_{a',x'}$ as follows:
\begin{align}
\sigma'_{a'|x'} &= \sum_{\lambda''} p(\lambda'') p(a'|x',\lambda'') \rho_{\lambda''},
\end{align}
where we have defined $p(\lambda'')= p(\lambda) p(\lambda')$,  $p(a'|x',\lambda'')=p(a'|a,x',\lambda') p(x|x',\lambda') p(a|x,\lambda)$ and $\rho_{\lambda''}=\mathcal E_{\lambda'}(\rho_\lambda)$. Note that the above decomposition is again an LHS model, but with a hidden variable dimension $d_{\lambda''} \le d_\lambda d_{\lambda'}$.  Thus, the new assemblage can have an LHS model with $d_{\lambda''} > d_A$, implying that its SEO can have noncommutativity.
\end{proof}

 Next, we proceed to prove Theorem~\ref{fo}. 
 From Lemma \ref{nFO}, it follows that the shared randomness $\lambda$ in the LOSR (\ref{FOinSEO}) can increase the dimension $d_\lambda$ of the given LHS assemblage. Therefore,  to define the free operations of SDI steering, we first uncorrelate $\lambda$ in Eq.~(\ref{FOinSEO}) by two independent sources of randomness $\mu$ and $\nu$, giving rise to transformations of the form given by
\begin{multline} 
\sigma'_{a'|x'}=\sum_{\mu,\nu,a,x} p(\mu)p(\nu) p(a'|a,x',\mu) p(x|x',\mu) \\
\cdot p(a|x) \mathcal E_\nu(\rho_{a|x}).
\end{multline}
Next, by defining the CPTP map $ \mathcal E=\sum_\nu p(\nu)  \mathcal  E_\nu$ and using $\sigma_{a|x}=p(a|x)\rho_{a|x}$, the above transformation is equivalent to the transformation of the form given by \begin{equation}\label{FOcSEOa}
\sigma'_{a'|x'}=\sum_{a,x,\mu} p(\mu) p(a'|a,x',\mu) p(x|x',\mu)
 \mathcal E(\sigma_{a|x}),
\end{equation}
as stated in Theorem~\ref{fo}.

We now demonstrate that operations of the form given by Eq.~(\ref{FOcSEOa}) cannot transform a state assemblage with commuting SEO into one with noncommuting SEO. Note that under any operation of the form~(\ref{FOcSEOa}), the assemblage having an LHS model with  $d_\lambda \le d_A$ transforms as follows:  \begin{align}
\sigma'_{a'|x'}=  &\sum_{a,x,\mu} p(\mu) p(a'|a,x',\mu) p(x|x',\mu)  \nonumber \\
&\cdot \mathcal E \left( \sum^{d_\lambda-1}_{\lambda=0} p(\lambda) p(a|x,\lambda) \rho_\lambda \right)    \\
= &\sum_{a,x,\mu} p(\mu) p(a'|a,x',\mu) p(x|x',\mu) \nonumber \\
&\cdot \sum^{d_\lambda-1}_{\lambda=0} p(\lambda) p(a|x,\lambda) \rho'_\lambda    \\
&=  \sum^{d_\lambda-1}_{\lambda=0} p(\lambda)  p(a'|x',\lambda) \mathcal \rho'_\lambda, \label{LTLHS}
\end{align}
 where in the second equality we have defined $\rho'_\lambda:=\mathcal E(\rho_\lambda)$ and in the last equality we have defined  $$p(a'|x',\lambda):=\sum_{a,x,\mu} p(\mu) p(a'|a,x',\mu) p(x|x',\mu)  p(a|x,\lambda).$$ Note that Eq. (\ref{LTLHS}) implies  another dimensionally-restricted LHS model for $\{\sigma'_{a'|x'}\}_{a',x'}$ with $d_\lambda \le d_A$. Therefore, under any local operation of the form (\ref{FOcSEOa}), any commuting SEO cannot be transformed into a noncommuting SEO.

\section{Proof of monotonicity of the measure in Definition~\ref{QSDIS}} \label{monotone}
First, we demonstrate the monotonicity of the measure of noncommutativity of any given SEO under free operations $\mathcal{F}$.
\begin{widetext}
\begin{align}
 &\Upsilon_{p}(\mathcal{F}[B_{a|x}]) \nonumber \\
 &= \sum_{a,a',a''.a''',x,x',\mu}||\left[ p(\mu) p(a''|a,x',\mu) p(x|x',\mu)
p(a|x) \mathcal E^\dag[B_{a|x}], \quad
p(\mu) p(a'''|a,x'',\mu) p(x|x'',\mu)
p(a|x) \mathcal E^\dag[B_{a'|x'}]\right]||_p \\
&\le \sum_{a,a'}||[\mathcal E^\dag[B_{a|0}],\mathcal E^\dag[B_{a'|1}]]||_p \\
&\le  \sum_{a,a'}||[B_{a|0},B_{a'|1}]||_p
=\Upsilon_{p}(B_{a|x}),
\end{align}
\end{widetext}
where in the second line we have used the triangle inequality for the Schatten norms $\sum_{a,b,i}||[p_i E_a,F_b]||_p \le \sum_i p_i \sum_{a,b} ||[E_a,F_b]||_p=\sum_{a,b} ||[E_a,F_b]||_p,$ for any two measurements and probability distribution $\{p_i\}_i$, with $p_i \ge 0$ and $\sum_i p_i=1$. In the third line, we have used that under CPTP maps $\mathcal E$:  $\mathbb{C}^{d} \rightarrow \mathbb{C}^{d}$ any free resource in the 1SSDI framework cannot be transformed into one with SDI steerability, in other words, noncommuting SEO cannot be created from commuting ones. The monotonicity of the measure of noncommutativity under $\mathcal{F}$ as shown above in turn implies that the measure of SDI steerability given by Definition~\ref{QSDIS} also has the same property.

\section{Illustration of the quantification SDI steerability}\label{illu}

As an illustration of the quantification of SDI steerability in Definition~\ref{QSDIS}, we apply it to specific two-qudit states.
First, we consider any pure entangled state $\ket{\phi^{(k)}_+}$ with the given Schmidt number $k$.
Suppose that Alice uses a measurement assemblage in two mutually unbiased bases (MUBs) given by  $M_{a|0}=\ketbra{a}{a}$, and $M_{a|1}=F\ketbra{a}{a}F^\dag$, where $F$ is the  discrete Fourier transform of $d$-dimension. Then, using Eq. ($8$) in the main text, the SEO is given by 
\begin{align}
\begin{split}
B_{a|0}&=\ketbra{a}{a}  \\
B_{a|1}&=F\ketbra{a}{a}F^\dag,
\end{split}
\end{align}
which implies that 
\begin{equation}\label{QSck}
S_{\Upsilon_p}(\sigma_{a|x}^{\ket{\phi^{(k)}_+}})=1,
\end{equation}
where $\{\sigma_{a|x}^{\ket{\phi^{(k)}_+}}\}$ is the state assemblage arising from $\ket{\phi^{(k)}_+}$ for the two MUBs.
Thus, all pure entangled states attain the same maximal value of the measure $S_{\Upsilon_p}(\sigma_{a|x})$. This is due to the concavity property of the noncommutativity measure.

Next, we consider the two-qudit isotropic state given by
\begin{equation}\label{iso}
\rho^{(d)}_{\rm iso}= \alpha \ketbra{\phi^{(d)}_+}{\phi^{(d)}_+} +\frac{1-\alpha}{d^2} \openone_{d^2}, \quad -\frac{1}{d^2-1} \le \alpha \le 1,
\end{equation}
which is the maximally entangled state $\ket{\phi^d_+}$ subjected to white noise. These states are separable for $ -\frac{1}{d^2-1} \le \alpha \le \frac{1}{d+1} $ and entangled otherwise~\cite{HH99}. For the two MUBs used in the case of pure entangled states above, the state assemblage arising from the two-qudit isotropic state satisfies 
\begin{equation}\label{SDISiso}
S_{\Upsilon_p}( \sigma_{a|x}^ {\rho^{(d)}_{\rm iso}})=|\alpha|,
\end{equation}
which is nonzero for any $|\alpha|>0$. 
Thus, even when the state is unsteerable in the standard steering context or separable, SDI steering is demonstrated for any $|\alpha|>0$.

Finally, we calculate $S_{\Upsilon_p}$ of $\rho^{(d)}_{\rm iso}$ using the same measurement as above, but with detection efficiency $\eta$. That is, instead of ideal measurements with elements $M_{a|x}$, we consider inefficient measurements $M^{(\eta)}_{a|x}$, with one additional outcome $a=\emptyset$,
given by
\begin{equation}
M^{(\eta)}_{a|x} = \Bigg\{\begin{array}{cc}
\eta M_{a|x}, & \qquad a \neq \emptyset \\
(1-\eta) \openone, & \qquad a = \emptyset.
\end{array}
\end{equation}
Now, using the two MUBs for $M_{a|x}$, $S^{(\eta)}_{\Upsilon_p}(\rho^{(d)}_{\rm iso})$ of $\rho^{(d)}_{\rm iso}$ in the presence of inefficient measurements is given by
\begin{equation}\label{SDISisoie}
S_{\Upsilon_p}( \sigma_{a|x}^ {\rho^{(d)}_{\rm iso}})=\eta|\alpha|,
\end{equation}
which implies that for any low detection efficiency $\eta$, $S_{\Upsilon_p}( \sigma_{a|x}^ {\rho^{(d)}_{\rm iso}})$ is nonzero for any $|\alpha|>0$. Whereas, for the steering demonstration in the 1SDI context, the detection efficiency must be greater than a certain threshold. For example, in the case of the maximally entangled two-qubit state, a genuine demonstration of steering requires that the detection efficiency be greater than $50\%$~\cite{PCS+15}, which may be too demanding for  experimental implementation~\cite{SVM+22}.

\section{Proof of Theorem~\ref{IRandom}}\label{app:thm3}
Consider that $\{p(a|x)\}_{a,x}$ is produced using a single quantum state $\rho_A$
and a measurement assemblage $\mathcal{M}^{n_a}_{n_x}$. In this case, the incoherent state $\rho_A=\sum_i p_i \ketbra{i}{i}$ in the basis, $\{\ket{i}\}$, and the uninformative measurement assemblage with the elements given by $M_{a|x}=p(a|x) \openone$ can produce randomness in any $\{p(a|x)\}$. Thus, global measurement on any single quantum state cannot be used to obtain intrinsic randomness. Next,  consider that $\{p(a|x)\}_{a,x}$ is observed using local measurements in the given 1SSDI scenario. Suppose that the given SEO has commutativity. Then  $\{p(a|x)\}_{a,x}$ can be produced using a classical-quantum (CQ) state $\rho_{\rm{CQ}}= \sum_i p_i \ketbra{i}{i} \otimes \rho^{(i)}_B $, where $\{\ket{i}\}$ forms an orthonormal basis and $p_i$'s are probabilities, for the measurement assemblage that produced the given state assemblage with commuting SEO. Thus, observing $\{p(a|x)\}_{a,x}$ in the presence of SEO that has commutativity does not imply intrinsic randomness, because it can be produced using an incoherent-quantum state, as mentioned above.

On the other hand, suppose $\{p(a|x)\}_{a,x}$ is observed in the given 1SSDI scenario with the SEO having noncommutativity. 
This implies that in the given 1SSDI scenario, the intrinsic randomness of $\{p(a|x)\}_{a,x}$ is certified by the noncommuting SEO because $\{p(a|x)\}_{a,x}$ cannot be produced using an incoherent-quantum state. In this context, suppose that Eve can hold a purification, $\ket{\psi_{ABE}}$, without any restriction on her dimension. Then, any correlation in the state assemblage $\mathcal{S}^{n_a}_{n_x}$ can be shared with Eve if entanglement is not certified~\cite{AG05,PCS+15,GBS16}. This follows because the state assemblage can be reproduced using a separable state, and the correlation in any separable state can be shared with Eve by the purification. Thus, Eve's dimension has to be restricted to that of Alice to certify randomness in $\{p(a|x)\}$ produced by any state assemblage with the SEO having noncommutativity.

\section{Proof of the inequality of quantum random number generation protocol} \label{proofSTQRNG}
Here, we provide the proof of the inequality given by Eq.~(\ref{guessub}) in the main text for the protocol of SDI quantum random number generation based on the witness of SDI steerability given by Definition~\ref{QSDIS}. The proof is inspired by that of the inequality used in~\cite{LBL+15} for the protocol of self-testing random number generation based on the prepare-and-measure setup.

First, considering the guessing probability $p^{(\mu)}_{g}$ for the given $\mu$, we have
\begin{align} \label{st1}
\begin{split}
p^{(\mu)}_{g}&= \max_a p(a|x,\mu), \\ 
             &\le \frac{1+\cos(\theta_\mu)}{2},  
\end{split}             
\end{align}
where $\theta_\mu$  is the angle between Alice's two measurements described by the Bloch vectors $\vec{T}^{(\mu)}_{0,1}$. 
The reasoning behind the above bound is as follows.  Let $p(a|x,\mu)$ be produced from a  state $\rho$ as
\begin{equation}
p(a|x,\mu)=\Tr (M^{(\mu)}_{a|x} \rho).
\end{equation}
 To obtain the best guessing probability $p^{(\mu)}_{g}$, optimizing $ p(a|x,\mu)$, given above for one of the outcomes, over all possible states $\rho$,  $p(a|x,\mu)=\cos^2(\theta_\mu/2)$, which is the upper bound in Eq. (\ref{st1}). This optimal value is achieved by a pure
 state that lies on the Bloch vector of the other measurement $\vec{T}^{(\mu)}_{x'}$.

Next, for fixed randomness $\mu$, consider the value of the witness denoted by  $S^{({\mu})}_\Upsilon$. To provide an upper bound on  $S^{({\mu})}_\Upsilon$, we make the following two observations on $S_\Upsilon$ of two-qubit states.
Let $\vec{r}$ and $\vec{v}$ denote the Bloch vectors of two qubit POVMs $B_{a|0}$ and $B_{a|1}$, respectively. Using this Bloch representation for the two POVMs,  we have $\sum_{a,a'}||[B_{a|0},B_{a'|1}]||_p=2^{\frac{1}{p}+1}||\vec{r}|||\vec{v}||\sin(\vec{r},\vec{v})$, which implies that
\begin{equation} \label{SBv}
S_{\Upsilon_p}=||\vec{r}|||\vec{v}||\sin(\vec{r},\vec{v}). 
\end{equation}
In Ref.~\cite{CS16}, it is known that an incompatibility quantifier of $M_{a|x}$ upper bounds a corresponding steering quantifier of any state assemblage that can be produced by Alice performing those measurements. On the other hand, we have the following  analogous relationship between the quantifier of the noncommutativity of $M_{a|x}$ and the noncommuativity measure of SDI steering,
\begin{equation}\label{SubMN}
S_{\Upsilon_p} \le \frac{1}{2^{\frac{1}{p}+1}} \Upsilon_p(M_{a|x})=\frac{1}{2^{\frac{1}{p}+1}} \sum_{a,a'}||[M_{a|0}, M_{a'|1}]||_p.
\end{equation}
Using Eqs.~(\ref{SBv}) and (\ref{SubMN}), we now have obtained an upper bound on $S^{({\mu})}_{\Upsilon_p}$ as follows:
\begin{equation} \label{st2}
S^{({\mu})}_{\Upsilon_p} \le  \sin\theta_\mu.
\end{equation}
Combining Eqs. (\ref{st1}) and (\ref{st2}), we get 
\begin{align} \label{st3}
p^{(\mu)}_{g}&= \frac{1}{2} \left(1+\sqrt{1-(S^{({\mu})}_{\Upsilon_p})^2}\right) =f(S^{({\mu})}_{\Upsilon_p}).
\end{align}

Now, we note the following convexity property of the witness $S_{\Upsilon_p}$,
\begin{equation} \label{st4}
S_{\Upsilon_p} \le \sum_\mu q_\mu S^{({\mu})}_{\Upsilon_p}.
\end{equation}
To prove this property, we rewrite the witness with internal randomness $\mu$ as follows:
\begin{align}
S_{\Upsilon_p}&=\frac{1}{2^{\frac{1}{p}+1}}\sum_{a,a',\mu}||q_\mu [B^{(\mu)}_{a|0},B^{(\mu)}_{a'|1}]||_p\\
 &\le \frac{1}{2^{\frac{1}{p}+1}}\sum_{a,a',\mu}q_\mu || [B^{(\mu)}_{a|0},B^{(\mu)}_{a'|1}]||_p\\
 &= \sum_\mu q_\mu S^{({\mu})}_{\Upsilon_p},
\end{align}
where in the second line  we have used the triangle inequality for the Schatten norms.

We can now proceed to obtain the upper bound on the guessing probability.  
Using the definition of $p_{g}$ together with Eq.~(\ref{st3}) and Eq.~(\ref{st4}), we have
\begin{align}
p_{g}&=\sum_\mu q_\mu p_{g}^{(\mu)}  \\
         &\le \sum_\mu q_\mu f(S^{({\mu})}_{\Upsilon_p}) \\
         &\le f(\sum_\mu q_\mu S^{({\mu})}_{\Upsilon_p}) \\
         &\le f(S_{\Upsilon_p}),
\end{align}
where in the third line we have used Jensen's inequality
and the concavity of $f$, and in the last line we have used that
$f$ is a decreasing function. Hence, finally, we get the desired inequality.

\bibliography{coherence}

\end{document}